\documentclass[10pt, conference]{ieeeconf}      

\IEEEoverridecommandlockouts                              
\newtheorem{theorem}{Theorem}

 \usepackage{graphics} 
\usepackage{epsfig} 
\usepackage{mathptmx} 
\usepackage{times} 
\usepackage{amsmath} 
\usepackage{amssymb}  
\usepackage{setspace}
\usepackage{color}
\usepackage{algorithm} 
\usepackage{algpseudocode} 
\title{\LARGE \bf
Privacy-Preserving Distributed Average Consensus in Finite Time using Random Gossip
}

\author{Nicolaos E. Manitara, 
Apostolos~I.~Rikos, 
and Christoforos N. Hadjicostis 
\thanks{}%
\thanks{The work of N. E. Manitara and C. N. Hadjicostis has been supported in part by the European Commission's Horizon 2020 research and innovation programme under grant agreement No 739551 (KIOS CoE). Any opinions, findings, and conclusions, or recommendations expressed in this publication are those of the authors and do not necessarily reflect the views of EC.}
\thanks{Nicolaos E. Manitara and C. N. Hadjicostis are with the Department of Electrical and Computer Engineering and with KIOS Center of Excelence, University of Cyprus, Nicosia 1678, Cyprus. Emails: {\tt \{manitara.nicolas, chadjic\}@ucy.ac.cy}}%
\thanks{Apostolos~I.~Rikos is with the Division of Decision and Control Systems, KTH Royal Institute of Technology, SE-100 44 Stockholm, Sweden. E-mail: {\tt rikos@kth.se}}
}

\begin{document}

\maketitle
\thispagestyle{empty}
\pagestyle{empty}

\begin{abstract}
In this paper, we develop and analyze a gossip-based average consensus algorithm that enables all of the components of a distributed system, each with some initial value, to reach (approximate) average consensus on their initial values after executing a finite number of iterations, and without having to reveal the specific value they contribute to the average calculation. We consider a fully-connected (undirected) network in which each pair of components (nodes) can be randomly selected to perform pairwise standard gossip averaging of their values, and propose an enhancement that can be followed by each node that does not want to reveal its initial value to other (curious) nodes. We assume that curious nodes try to identify the initial values of other nodes but do not interfere in the computation in any other way; however, as a worst-case assumption, curious nodes are allowed to collaborate arbitrarily and are assumed to know the privacy-preserving strategy (but not the actual parameters chosen by the nodes that want to preserve their privacy). We characterize precisely conditions on the information exchange that guarantee privacy-preservation for a specific node. The protocol also provides a criterion that allows the nodes to determine, in a distributed manner (while running the enhanced gossip protocol), when to terminate their operation because approximate average consensus has been reached, i.e., all nodes have obtained values that are within a small distance from the exact average of their initial values.
\end{abstract}

\section{Introduction}\label{sec:introduction}
In distributed systems and networks, it is often necessary for all or some of the components (nodes) to calculate a function of certain parameters that we refer to as initial values. When all nodes calculate the average of their initial values, they are said to reach average consensus. Average consensus (and more generally consensus) has received a lot of attention from the control community due to its usage in various emerging applications, including wireless smart meters (where all nodes have to agree on the average power demand or consumption of the network \cite{2011:Christoforos_Garcia}), and multi-agent systems (where all agents communicate with each other in order to coordinate their direction or speed \cite {2005:Ren_Beard_Atkins}). Over the last few decades, a variety of distributed algorithms for calculating different functions of these initial values have been proposed by the control, communication, and computer science communities \cite{1996:Lynch,2003:Koetter,2005:Hromkovic,2008:Cortes}.

One approach to consensus is based on a linear iterative strategy, where each node in the network repeatedly updates its value to be a weighted sum of its own previous value and the values of its neighbors. If the network topology satisfies certain conditions, the weights for the linear iteration can be chosen such that all of the nodes in the network converge (asymptotically) to the same function of the initial values, which under additional conditions on the weights, can be the average \cite{2008:Sundaram_Hadjicostis, 2018:BOOK_Hadj}. Another popular approach to consensus is a gossip-based strategy \cite{2010:Dimakis_Rabbat} which allows the nodes to calculate the average value of the network without having to follow a predefined routing of the information that needs to be transmitted/received to/from other nodes in the network. At each iteration, information is exchanged between a randomly selected pair of neighboring nodes, and then this information is processed by both nodes to compute a local pairwise average. 

Both of the methods described above are asymptotic and typically do not consider privacy issues that may arise while the network is reaching consensus. This paper addresses the topic of privacy-preserving average consensus in finite time, which has received some recent attention. Specifically, an anonymization transform using random offsets on the initial values was proposed for a cooperative\footnote[1]{In a cooperative network, all the nodes follow the predefined strategy, without deviating in any way \cite{Kefayati:2007}.} wireless network in \cite{Kefayati:2007}. In this approach, each node that would like to protect its privacy (i.e., does not want to reveal its initial value) adds a random offset value to its initial value, ensuring in this way that its true initial value will not be exposed through the values exchanged in the network. Similarly, in the privacy-preserving strategies proposed in \cite{Nozari_2017} nodes inject uncorrelated noise into the exchanged messages so that the data associated to a particular node cannot be inferred by a curious node during the execution of the algorithm. 
However, note that the average value is not obtained exactly in \cite{Nozari_2017} leading to notions of differential privacy \cite{2016:CortesPappas}. 
The injection of correlated noise at each time step and for a finite period of time was proposed in \cite{2013:Nikolas_Hadj} and guarantees convergence to the exact average. 
In \cite{Mo-Murray:2017}, the nodes asymptotically subtract the initial offset values they added in the computation, while in \cite{west} each node masks its initial state with an offset such that the sum of the offsets of each node is zero, thus guaranteeing convergence to the average. 
In \cite{2020:RikosThemisJohHadj_CDC, 2021:Priv_Rikos_Hadj_Joh_Themis_TCNS}, the nodes inject an initial quantized offset to their initial values. 
During the operation of the algorithms, the nodes are able to subtract the initial offset values they added, which guarantees convergence in finite time to the exact average. 
In \cite{2019:Cascudo_Christensen} the nodes achieve privacy-preserving distributed average consensus using the principle of additive secret sharing where each node replaces its initial value with another obfuscated value by subtracting and adding random numbers.
Finally, another approach that guarantees privacy-preservation is via homomorphic encryption \cite{Hadjicostis_Garcia_2020, 2019:Ruan_Wang, Hadjicostis_2018}. 
However, this approach requires the existence of trusted nodes and imposes heavier computational requirements on the nodes. 

Our own previous work in \cite{2013:Nikolas_Hadj} describes a privacy-preserving protocol that is a variation of the standard linear iteration in \cite{2008:Sundaram_Hadjicostis} that is used in the absence of privacy requirements under undirected communication topologies using weights that form a doubly stochastic matrix. The main enhancement is that, at each time-step, each node following the protocol adds an arbitrary offset value and/or the result of its iteration, in an effort to avoid revealing  its own initial value or the initial values of other nodes. What is important is for each node to ensure that the total offset it adds cancels out in the end (the accumulated sum of offsets is zero). The work in \cite{2013:Nikolas_Hadj} establishes that under certain conditions on the communication topology, the protocol allows the nodes to calculate the average of their initial values in a privacy-preserving manner, despite the presence of curious agents. For directed topologies, inspired by ratio consensus strategies, we devised in \cite{2019:Charalambous_Manitara_Hadjicostis} a distributed mechanism, where each node updates its information state by combining the available information received from its in-neighbors using constant positive weights and by adding an offset to one of the two states communicated during the execution of the algorithm. As shown in \cite{2019:Charalambous_Manitara_Hadjicostis}, if after a finite number of iteration the nodes ensure that the total offset is zero, this privacy-preserving version of ratio consensus converges to the exact average of the nodes’ initial values, even in the presence of bounded time-varying delays.

The privacy-preserving gossip-based protocol we develop and analyze in this paper enables all of the nodes to calculate in finite time an $\epsilon$-close approximation of the average of their initial values, without loss of privacy despite the presence of multiple (possibly colluding) curious nodes. Curious nodes are assumed to have full knowledge of the proposed protocol and are allowed to collaborate arbitrarily among themselves, but do \emph{not} interfere in the computation of the average value of the network in any other way. Our approach does not depend on any cryptographic algorithm, but operates by allowing the nodes to introduce pseudo-random offsets (unknown to the curious nodes). Specifically, the proposed protocol enhances the standard gossip-based protocol \cite{2010:Dimakis_Rabbat} in two main directions. The first enhancement is that, at each time-step, the randomly chosen pair of nodes not only exchanges information and calculates the local (pairwise) average value, but if any of the two is nodes following the protocol, then that node also adds an arbitrary offset value to its updated (average) result, in an effort to avoid revealing its own initial value or the initial values of other nodes. What is important is for each node to ensure that the total offset added (accumulated sum of offsets) eventually becomes zero. We establish simple conditions to ensure that the proposed protocol allows the nodes to calculate the average of their initial values in a privacy-preserving manner, despite the presence of curious agents. In particular, the protocol guarantees that when there are at least two nodes following the protocol then privacy is guaranteed for both of them in the sense that their individual initial values are not revealed to the curious nodes. Note that it might still be possible for the curious nodes to determine exactly the sum of the initial values of nodes that follow the privacy-preserving protocol (but not their individual values). The second enhancement is a deterministic criterion that enables the nodes to determine, in a distributed manner, when to seize initiating gossip exchanges because approximate average consensus has been reached, i.e., it is shown that eventually all nodes will seize activity, and that point they will all have values that are within a small distance from the true average of their initial values.

The remainder of the paper is organized as follows. In Section~\ref{sec:background}, we provide some relevant background on graph theory and describe the linear and gossip-based iterative strategies for reaching average consensus in a given distributed system. In Section~\ref{sec:Problem Setup} we formulate the problem of interest. We describe our proposed privacy-preserving gossip protocol, and  the main results of the paper in Section~\ref{sec:Proposed_Strategy_Main_Results}. In Section~\ref{sec:Example} we present an example, and finish the paper with conclusions and directions for future work in Section~\ref{sec:conlcusion}.

\section{Background}\label{sec:background}
\subsection{Distributed System Model}

 \par In a distributed system we can model the network topology as a directed graph (digraph) \textit{G=\{X,E\}} where \textit{X=\{$1,2,...,N$\}} is the set of components in the system and  \textit{$E\subseteq X \times X-\{(i,i)\mid {i\in X}$}\} is the set of edges (self-edges excluded). 
 In particular, edge $(i,j)\in E$ if node $j$ can send information to node $i$. In this work, we focus on components that interact via bidirectional links in a way that forms a connected undirected (or symmetric) graph, i.e., a graph for which $(i,j)\in E$ if and only if  $(j,i)\in E$. The set of nodes that can exchange information with node \textit{$i$} are said to be the neighbors of node \textit{$i$} and are represented by the set $\mathcal{N}_i=\{(j,i)\in E\}$. 
 The number of neighbors of node \textit{i} is called the degree of node \textit{$i$} and is denoted as $\mathcal{D}_i$=$|\mathcal{N}_i|$.
 
 Our model deals with networks where information is transmitted/received between a randomly selected pair of nodes at each time step. More specifically, each node randomly wakes up and initiates an information exchange with a randomly selected neighbor. We assume that nodes are aware of the unique identity (id) of each node they communicate with and that, during the exchange process, all the information to/from a particular node is transmitted/received successfully \cite{1996:Lynch}. 
 Moreover, the nodes must have sufficient memory and computational capability in order to store and perform certain computations while the iteration is executing. 
 During the transmission/reception process, the nodes in the network receive a value from the selected node, and transmit their value back to the selected node.

\subsection{Average Consensus via Linear Iterative Strategy}

 In average consensus problems the objective is the calculation of the average of the initial values of the nodes in the network. 
 Assume that each node \textit{$i$} has some initial value   \textit{$x_i$} and, at each time-step $k$, it updates its value as a weighted sum of its own value and the values of its in-neighbors (e.g., following the method in \cite{2008:Sundaram_Hadjicostis}). 
 Specifically, at each time-step $k$, each node updates its value as
 \begin{eqnarray}\label{aver_iteration}
 x_i[k+1]=w_{ii}x_i[k]+\sum_{j\in \mathcal{N}_i} w_{ij} x_j[k],
\end{eqnarray}
where $w_{ij}$ are a set of (fixed) weights and $x_i[0]=x_i, \forall i\in X $. 
The values for all the nodes at time-step $k$ can be aggregated into the value vector $x[k]=\big[x_1[k]$ $x_2[k]$ $ ...$ $x_N[k]\big]^T$ and the update strategy for the entire network can be written compactly as
\[ x[k+1]=\underbrace{\left[
\begin{array}{cccc}
 w_{11} &w_{12} &\cdots &w_{1N} \\
 w_{21} &w_{22} &\cdots &w_{2N} \\
 \vdots  &\vdots  &\ddots &\vdots  \\
 w_{N1} &w_{N2} &\cdots &w_{NN}
\end{array} \right]}_W x[k], \]for $k\in\mathbb{N}$, where $w_{ij}=0$ if $x_j \notin \textit{$\mathcal{N}_i$} \cup {\{i\}}$.

\par$  $
The system is said to reach asymptotic consensus if $lim_{k\rightarrow\infty}x_i[k]=\textit{f}(x_1[0],x_2[0],...,x_N[0])$ for each \textit{i}, where $\textit{f}(x_1[0],x_2[0],...,x_N[0])\in \mathbb{R}$. When $\textit{f}(x_1[0],x_2[0],...,x_N[0])=c^Tx[0]$ for some column vector $c$ (where $c^T$ is the transpose of $c$), the following result by Xiao and Boyd from \cite{2004:XiaoBoyd} characterizes the conditions under which iteration \eqref{aver_iteration} achieves asymptotic consensus.
\begin{theorem} \cite{2004:XiaoBoyd} The iteration given by \eqref{aver_iteration} reaches asymptotic consensus on the linear functional $c^T x[0]$ (under the technical condition that $c$ is normalized so that $c^T \bf{1}=1$, where $\bf{1}$ is the all ones column vector of size $N$) if and only if the weight matrix $W$ satisfies the conditions below:
\begin{enumerate}
 \item $W$ has a simple eigenvalue at 1, with left eigenvector $c^T$ and right eigenvector $\bf{1}$=[1 1...1$]^T$;
   \item All other eigenvalues of $W$ have magnitude strictly less than 1.
   \end{enumerate}
In particular, if $c^T$=$\frac{1}{N}$[1 1...1], then average consensus is reached. Also note that if $w_{ij}$ are restricted to be nonnegative, then the above conditions are equivalent to $W$ being a doubly stochastic matrix.
\end{theorem}

\subsection{Average Consensus using a Randomized Gossip Strategy}
 In a randomized gossip protocol, a pair of nodes is randomly selected at each iteration and perform local (pairwise) averaging of the two values. 
Specifically, at iteration $k$, nodes $i$ and $j$ (if selected) update their values as
\begin{eqnarray}
\label{EQgossip}
x_i[k+1]=x_j[k+1]=\frac {x_i[k]+x_j[k]}{2},
\end{eqnarray} with $x_i[0]=x_i$~for all $i$. Typically, it is assumed that each node ``wakes up'' in an asynchronous manner and randomly selects a neighbor to perform local averaging with. Note that the remaining nodes do not update their values, i.e., $x_l[k+1]=x_l[k]$ for all nodes $l$, $l\neq j$ and $l \neq i$. 

It can be shown that, if nodes are following the above protocol, they are guaranteed to asymptotically reach the average of their initial values $\bar x=\frac{1}{N}\sum_{l=1}^{N} x_l[0]$, as long as the network is connected (so that all the available information can flow to every node in the network) and all pairs of nodes in the network are active frequently enough. 

\subsection{Previous Work on Privacy-Preserving Average Consensus}
\par Privacy-preserving average consensus in the presence of  curious agents in the network has already received significant attention in the literature (e.g., \cite{Kefayati:2007}, \cite{2013:Nikolas_Hadj}, \cite{2010:Dimakis_Rabbat}, \cite{Mo-Murray:2017} \cite{2016:CortesPappas}). Below, we briefly review the most relevant existing results for our work here.

As mentioned earlier, the authors of \cite{Kefayati:2007} proposed a transformation method using random offset values in a cooperative wireless network. Specifically, each node \textit{$i$} that wishes to protect its privacy adds a random offset value \textit{$u_i$} to its initial value \textit{$x_i$}. This ensures that its value will not be revealed to curious nodes that might be observing the exchange of values in the network. The idea is based upon the observation that, when a large number of nodes employ the protocol, their offsets will have a zero net effect on the average with high probability, allowing the nodes to converge to the true average of the initial values of the nodes in the network. Specifically, each node \textit{$i$} sets $x_i[0]$=\textit{$x_i$}$'$= \textit{$x_i$+$u_i$} where \textit{$u_i$}, \textit{$i=1,2,...,N,$} are i.i.d. random variables with zero mean. Then, following the protocol for asymptotic average consensus, the nodes converge to, \begin{eqnarray}
\frac{1}{N} \sum_{i=1}^{N} x_i'=\underbrace{\frac{1}{ N} \sum_{i=1}^{N} x_i}_{\overline{x}}+\underbrace{\frac{1}{N} \sum_{i=1}^{N} u_i}_U,
 \end{eqnarray}where $\overline{x}$ is the desirable average of the original initial values and $U$ is a random variable that captures the net effect of the offsets.
Clearly, $E[U]=0$ (since the $u_i$ are zero mean) and $var[U]=\frac{1}{N}var(U_i)$ (since the $u_i$ are i.i.d.).

For ${N\to \infty}$ we have $var(U)~\rightarrow 0$ which means that
\begin{eqnarray}
 \lim_{N \to \infty} \frac{1}{ N} \sum_{i=1}^{N} x'_i = \lim_{N \to \infty} \bigg[\frac{1}{ N} \sum_{i=1}^{N} x_i + \frac{1}{ N} \sum_{i=1}^{N} u_i \bigg] =\overline{x}+0=\overline{x},
 \end{eqnarray}
 where the last inequality is taken in the mean square sense.
For a large number of nodes $(N\rightarrow\infty)$, this method can give results very close to the true average of the network; however, as the number of nodes decreases, the accuracy of this method also decreases, due to the fact that the offset values added to the protected nodes will add a random offset (with mean zero and some finite variance) to the true average value of the network.

In our own previous work \cite{2013:Nikolas_Hadj}, the objective is also to calculate the average of the initial values of the nodes in the network, while at the same time preserving the privacy of the nodes following the protocol. The scheme that is used in \cite{2013:Nikolas_Hadj} assumes that the underlying network forms a connected undirected graph and makes use of linear iterations as in \eqref{aver_iteration} were the weights $w_{ij}$ form a doubly stochastic matrix $W=[w_{ij}]$ (thus, the nodes asymptotically reach consensus to the average of their initial values). The main difference is that node $i$ following the protocol sets its initial value $x'_i[0]=x_i+u_i$ (where $u_i$ is some random offset), and subsequently updates its value as
 \begin{eqnarray}
 x_i'[k+1]=w_{ii}[k]x_i'[k]+\sum_{j\in \mathcal{N}_i} w_{ij} x_j'[k] + u_i[k],\ k=0,1,...,
\end{eqnarray}
where $u_i[k]$ is a pseudo-random value chosen by node $i$ at time-step $k$. The constraint is that $u_i[k]=0$ for $k>L_i$ (for some $L_i$ known only to node $i$) and 
\begin{eqnarray}\label{last_offset}
u_i[L_i]=-\sum_{t=0}^{L_i-1}u_i[t]-u_i.
\end{eqnarray}
At time-step $L_i$, node $i$ effectively cancels-out the pseudo-random values it has added during the information exchange in the network up to that point.

The main contribution of \cite{2013:Nikolas_Hadj} is the establishment of topological conditions that ensure privacy for the nodes following the proposed protocol despite the presence of curious nodes in the network. Considering a fixed network with $N$ nodes described by a digraph $G=\{X,E\}$ and the iteration in \eqref{aver_iteration} with weights that form a primitive doubly stochastic weight matrix $W$, it assumes that a set of nodes $P$ follow the predefined privacy-preserving strategy in \eqref{EQgossip} with random offsets chosen as in \eqref{last_offset}. Curious node $i$ will not be able to identify the initial value $x_j$ for $j\in{P}$, as long as $j$ has at least one other node $l$ connected to it for which all paths from $l$ to the curious node $i$ are through a node $j'$ following the protocol (i.e., $j'\in P$). Specifically, if the above mentioned condition is satisfied, the network will reach average consensus and the privacy of the initial values of the nodes following the protocol will be preserved during the linear iteration process.

\section {Problem Setup}\label{sec:Problem Setup}
Consider a network with $N$ nodes each with some initial value (node $i$ has initial value $x_i$). 
Let us assume that all nodes follow a gossip-based strategy for asymptotically reaching agreement to the average of their initial values. 
However, we consider that (i) a set of curious nodes $X_c \subset X$ try to identify the initial values of all or some of the nodes in the network, (ii) a set of private nodes $X_p \subset X$ would like to preserve their privacy by not revealing to other nodes their initial values, and (iii) a set of neutral nodes $X_n \subset X$ are neither curious nor try to preserve the privacy of their initial values (we allow some nodes not to be curious and not follow the privacy-preserving protocol in order to investigate the worst-case scenario that this protocol can handle). 
In this paper we describe a distributed protocol that allows each node $i$ to calculate the average $ {\overline x}$, while every node $i \in X_p$ is able to preserve the privacy of its initial value from possibly colluding curious nodes in $X_c$. 
Furthermore, once values close to $\overline{x}$ obtained for all nodes, each node is able to determine when to terminate its operation in a distributed manner using a deterministic criterion. 
More specifically when the nodes terminate their operation, we are guaranteed that approximate average consensus has been reached, i.e., all nodes have obtained values that are within a small distance from $\overline{x}$.
Note that in this paper, we assume that the curious nodes have full knowledge of the proposed protocol and are allowed to collaborate arbitrarily among themselves (exchanging information as necessary), but do not interfere in the computation of the average value in any other way. 

\section{Proposed Strategy and Main Results}\label{sec:Proposed_Strategy_Main_Results}

The main contribution of this paper includes a deterministic privacy-preserving gossip protocol that also guarantees distributed stopping. The protocol is deterministic in the sense that, at its completion, it is guaranteed that all values are close to $\overline{x}$ and the nodes that follow the privacy-preserving protocol have protected their initial values from exposure, as long as a simple sufficient condition is satisfied.

\subsection{Deterministic Privacy-Preserving Protocol}\label{Priv_prot}
The first objective of the protocol is to calculate the average of the initial values of the nodes in the network while preserving the privacy of the nodes that follow the proposed protocol. The scheme that we study makes use of the iteration in \eqref{EQgossip} where, at each iteration, a randomly selected pair of nodes perform local averaging of their values and asymptotically reach approximate consensus to the average of their initial values $\overline x$, at least as long as each pair of nodes is selected enough times. What we propose in this work is a deterministic privacy-preserving protocol in which node $i$ following the protocol sets its initial value to
 \begin{eqnarray}\label{initial_injection}
x'_i[0]=x_i+u_i
\end{eqnarray}
where $u_i$ is some random offset. A node $l$ that is not following the protocol simply sets $x'_l[0]=x_l$. Subsequently, when nodes $i$ and $j$ are selected to perform local averaging (say at iteration $k$), node $i$ updates its value as
 \begin{eqnarray}\label{iteration_injection}
 x_i'[k+1]=\big(x_i'[k]+x_j'[k]\big)/2 + u_i[k],\ k=0,1,...,
\end{eqnarray}
where $u_i[k]$ is a pseudo-random value chosen by node $i$ at iteration $k$. Node $j$ does something similar (adding a pseudo-random offset $u_j[k]$) if it is following the protocol; if node $j$ is not following the protocol, then node $j$ sets
 \begin{eqnarray} \label{iteration_noinjection}
 x_j'[k+1]=\big(x_i'[k]+x_j'[k]\big)/2,\ k=0,1,... .
\end{eqnarray}

The constraints are (i) $u_i[k]=0$ for $k>L_i$ where $L_i$ is the time step where node $i$ is first selected after it identifies that it has exchanged information with each \textit{node} in the network at least once, and (ii) \begin{eqnarray}u_i[L_i]=-\sum_{t=0}^{L_i-1}u_i[t]-u_i.\end{eqnarray}
Note that at time-step $L_i$, node $i$ effectively cancels-out the pseudo-random values it has added during the information exchange in the network up to that point; thus, from that point onward the remaining accumulated offset (due to node $i's$ offset values) in the system is zero. Note that for notational convenience we take $u_i[k]=0$ for all time steps $k$ at which node $i$ is not selected to update its value.



\textit{Protocol Description:} We now describe the distributed protocol from the perspective of node $i$. When selected, nodes following the protocol run the linear iteration in \eqref{iteration_injection} or \eqref{iteration_noinjection} (depending on whether they follow the protocol or not). Each time a link is activated, each node among the selected  pair of nodes is aware of the identity of the other node involved; thus, nodes are able to identify when they have communicated at least once with all the nodes in the network. We refer to $L_i$ as the iteration ($k=L_i$) at which node $i$ is first selected after it has communicated at least once with all other nodes. A node not following the protocol sets $u_i=0$ and $u_i[k]=0$ for $k=0, 1,2,...,$ which is the standard gossip protocol for reaching average consensus. Node $i$ that follows the privacy-preserving protocol executes \eqref{iteration_injection}, with $x'_i[0]=x_i+u_i$ and
\begin{itemize}
   \item[i)] Chooses a pseudo random offset $u_i[k], \ k=0,1,...,L_i-1$ for integer $L_i$, ($u_i[t]=0$ for iteration steps $t$ for which node $i$ is not selected);
   \item [ii)]Sets \begin{eqnarray} \label{EQoffset} u_i[L_i]=-\sum_{t=0}^{L{_i}-1}u_i[t]-u_i;\end{eqnarray}
   \item [iii)]Sets $u_i[k]=0$ for $k > L_i$.
 \end{itemize}
Note that $L_i$ is an integer number of steps that becomes known to each node \textit{i} following the protocol. This number is a random variable since the pairs are selected randomly at each iteration.
\par $Lemma~1$: Following the iteration in \eqref{aver_iteration} and in combination with the constraint in \eqref{EQoffset} the network will reach asymptotic average consensus.

\begin{proof}
 It is not hard to see that, if we let $L_{max}=max_i\{L_i\}$, then    \begin{eqnarray}\sum_{i=1}^N x'_i[L_{max}+1]=\sum_{i=1}^N x_i;\nonumber\end{eqnarray}  then, using the fact that from this point onward nodes follow a standard gossip protocol, we conclude that
\begin{eqnarray}
 \lim_{k \to \infty} x'_i[k]=\frac{1}{ N} {\sum_{i=1}^N x'_i[L_{max}+1]}=\frac{1}{ N} {\sum_{i=1}^N x_i}=\overline{x},\nonumber
 \end{eqnarray}for all~$i\in X$ (i.e., average consensus has been reached).\end{proof}

\begin{algorithm}[h!]
1:~Consider a fully connected graph $G=\{X,E\}$ with $N=|X|$ nodes. Each node $i$ has initial value as $x_i$.\\
2:~If node $i$ is following the protocol, it sets its initial value to $x'_i[0]=x_i+u_i$ (where $u_i$ is some random offset); otherwise, $x_i'[0]=x_i$.\\
3:~Let $L_i$ be the time step at which node $i$ is next activated after it communicates at least once with all the other nodes in the network). During the first $L_i-1$ steps, if node $i$ is selected, it adds a random offset value to its transmitted value. At time-step $L_i$, node $i$ cancels out the injected noise in the system by setting $u_i[L_i]=-\sum_{t=0}^{L{_i-1}}u_i[t]-u_i$, and for $k>L_i,$ it sets $u_i[k]=0$.\\
4:~Specifically, for $k=0,1,2,...,$ each node $i \in X$, if selected with node $j$, does the following ($u_i[k]=0$ for all $k$ if $i$ is not following the protocol)\\ 
5:~\textbf{if $k< L_i$ then:}~$x_i'[k+1]=\frac{(x_i'[k]+x_j'[k])}{2}+u_i[k]$\\
6:~\textbf{if $k=L_i$ then:}~$x_i'[k+1]=\frac{(x_i'[k]+x_j'[k])}{2}-\sum_{t=0}^{L{_i-1}}u_i[t]-u_i$\\
7:~\textbf{if $k>L_i$ then:}~$x_i'[k+1]=\frac{x_i'[k]+x_j'[k]}{2}$\\
8:~End
	\caption{Privacy-Preserving Gossip Protocol} 
\label{algorithm:1}
\end{algorithm}
 
\subsection{Deterministic Distributed Stopping Protocol}\label{stop_prot}
In the proposed deterministic distributed stopping algorithm, each node makes a decision on how to update and/or transmit its value, based on the difference between its calculated value and the value
it receives from its selected neighbor at each iteration that the node is activated. In the algorithm, the iterative process ends when all nodes cease transmitting their values, in which case they can be shown to have reached approximate average consensus, i.e., the absolute difference between the final value of each node and the exact average $\overline{x}$ of the initial values is smaller than an error bound $\epsilon$ (small real value).
\\ In terms of the definition below we are interested in reaching $\varepsilon$-approximate average consensus and also in identifying (in a distributed manner) when such approximate average consensus has been reached.
\par $Definition~2~(\varepsilon-Approximate~Average~Consensus):$ At the end of the iterative process, the nodes have reached $\varepsilon$-approximate average consensus if the value $x_i[f]$ of each node $i \in X$ satisfies $| x_i[f]- \bar x | \leq \varepsilon \; ,  \; \forall i \in X$, where $\bar{x}$ is the average of the initial values.

\textit{Protocol Description:} We now describe the distributed stopping protocol from the perspective of node $i$. Nodes following the protocol run the linear iteration in \eqref{iteration_injection} and \eqref{iteration_noinjection},  depending on whether or not they  run the privacy-preserving protocol, in order to reach average consensus. Nodes running the privacy-preserving protocol wait until they communicate at least once with each other node and then start executing the distributed stopping protocol below. The protocol has two phases. During phase 1, each time node $i$ is activated (say with node $j$) checks whether the absolute difference $|x_i[k]-x_j[k]| < \varepsilon$; if the condition is satisfied then the flags $f_{ji}$ and $f_{ij}$  that represent the closeness of the values between the two nodes become $1$; if the condition is not satisfied then the nodes proceed to perform local averaging of their values and set all their flags to zero. Once node $i$ has all its flags $\{f_{ji}\mid j\in X\}-\{f_{ii}\}$ equal to $1$, it enters the second phase. In this phase, node $i$ does not initiate any exchanges but, if it is selected by another node $j$, it checks whether the absolute difference $|x_i[k]-x_j[k]| > \varepsilon$ holds; in such case, all the flags of node $i$ become $0$ and node $i$ goes back to phase 1. 

\begin{algorithm}[h!]
\caption{Distributed Stopping for Average Consensus}
\begin{algorithmic}[1]
\\
Input: Each node $i$ initializes $x_i[0]=x_i$ and if node $i$ follows the privacy protocol, then it starts to implement the distributed stopping protocol described below, {\em after} time step $L_i$. 
\\
Initialize all flags $\{f_{ji}\mid j\in X-\{i\}\}$ to zero.\\
For each  $k>L_i$ when a pair of nodes $i$ and $j$ is randomly selected, nodes $i$ and $j$ do the following.\\
~~If $\mid x_i[k]-x_j[k]\mid<\varepsilon$\\
~~~~$Flags_{ij}=Flags_{ji}=1$  \\
~~else\\
~~~~$x_i[k+1]=x_j[k+1]=\frac{x_i[k]+x_j[k]}{2}$\\
~~~~$Flags_{*i}=Flags_{*j}=0$~~~~~(set all flags to zero)\\
~~End\\
~If all $f_{ji}=1$ node $i$ does not initiate any gossip interactions (unless contacted by another node $j$).\\
~~End\\
End
\end{algorithmic}
\label{algorithm:2}
\end{algorithm}

\begin{theorem}\label{theorem_convergence}
Consider a network described by a fully connected undirected graph $G=\{X,E\}$, where nodes run the enhanced gossip algorithm in order to reach average consensus. Following the deterministic distributed stopping protocol, nodes reach $\varepsilon$-approximate average consensus after a finite number of iterations.
\end{theorem}

\begin{proof}
Due to space limitations, we only provide a sketch of the proof. We can first establish, by contradiction, that the nodes will stop after a finite number of iterations $f$. Suppose that the iteration runs forever; this means that at each iteration $k$ at least two neighboring nodes are active (transmitting nodes).  Let $A=\{(i_1,j_1), (j_1,i_1), (i_2,j_2), (j_2,i_2), \ldots, (i_\kappa, j_\kappa), (j_\kappa, i_\kappa)\}\subseteq E$ be the {\em non-empty} set of edges that are active infinitely often. Also, let $\ell$ be the latest iteration at which an edge in the set $E-A$ is active (note that $\ell$ is finite). We have that the active edge at iteration $k>l$ denoted by $e[k]$ satisfies $e[k]\in A$ (i.e., after time step $\ell$ each active edge belongs in the set $A$).
Since each $e[k]$ results in an undirected graph $G[k] = \{ X, \{e[k\} \}$, the graph $\{ X, A \}$ is also an undirected graph and can be partitioned into connected components. Then, from \cite{west}, it follows that the nodes in each connected component asymptotically reach the same value. This implies that the edges in the connected component would cease to be active, which is a contradiction.
Since the nodes stop at some iteration $f$, we have $\varepsilon$-approximate local consensus \cite{Allerton13_Manitara}. Using  Proposition~2 from \cite{Allerton13_Manitara}, we also have $(D \times \varepsilon)$-approximate global consensus, and since we have a fully connected graph $(D=1)$, the absolute difference between any pair of node values after they stop transmitting is at most $\varepsilon$. Moreover, we have
$$
x_{\min}[f] \leq
\frac{1}{N}\sum_{l=1}^{N} x_l[f]
\leq x_{\max}[f]
$$
(where $x_{\min}[f] = \min_i \left \{x_i[f] \right \}$ and $x_{\max}[f] = \max_i \left \{x_i[f] \right \}$), which implies that
$$
x_{\min}[f] \leq \bar{x} = \frac{1}{N} \sum_{l=1}^N x_l=\frac{1}{N} \sum_{l=1}^N x_l[f] \leq x_{\max}[f] \; .
$$
Since $\mid x_{max}[f]-x_{min}[f] \mid < \varepsilon$ this completes the proof. 
\end{proof}

\subsection{Conditions for Privacy-Preservation}

\begin{theorem}\label{privacy_theorem}
Let us consider a fully connected graph $G=\{X,E\}$ of $N$ nodes. 
The set of private nodes $X_p$ follow the predefined privacy-preserving gossip strategy in \eqref{EQgossip} with random offsets chosen as in \eqref{iteration_injection} and \eqref{EQoffset}. 
The set of curious nodes $X_c$ will not be able to identify the initial value $x_j$ of node $j \in X_p$, as long as 
\begin{enumerate}
    \item there is at least one node $l \in X_p$ in the network that also follows the proposed privacy-preserving protocol; 
    \item there is at least one node $l \in X_n$ in the network which first exchanges information with node $j$ (at the first iteration it is involved in a selected pair). 
\end{enumerate}
\par Note that, if the conditions in Theorem~$2$ are also satisfied, the network will reach approximate average consensus in finite time, while the privacy of the initial values of the nodes following the protocol will be preserved during the process.
\end{theorem}

\begin{proof}
The proof of Theorem~\ref{privacy_theorem} analyzes the following simple scenarios. 
\begin{itemize}
\item[1)] If all nodes $i \in X - \{ j \}$ in the network are curious and they communicate with each other, it is not possible for $j$ to maintain the privacy of its initial value $x_j$. 
Specifically, at the initialization of our algorithm, the curious nodes will know the value $x_j + u_j$ (recall that $u_j$ is the value that $j$ injected to its initial value in \eqref{initial_injection}). 
During the iteration of our algorithm, curious nodes in $X_c$ will get to learn the values $u_j[k]$ for $k = 0, 1, ..., L_j$, that $j$ injects to $x_j'[k]$ (since curious nodes collude with each other and exchange information, they can compare each pair of consecutive values node $j$  announces to determine $u_j[0]$, then $u_j[1]$, and so forth up to $u_j[L_j]$). 
This means that the curious nodes $i \in X_c$ will be able to compute the total amount of injected offsets $u_j + \sum_{t = 0}^{L_j-1} u_j[t] = - u_j[L_j]$ and, from that, $u_j$ and therefore the initial value $x_j$ of node $j$. 
As a result, the privacy of node~$j$ is not preserved and at least one node (other than $j$) that is not curious is needed in the network. 
It is also worth highlighting that no privacy-preserving protocol can protect node $j$ in this case (since $\overline{x}$ becomes known and the curious nodes know their own value).
\item[2)] Let us consider the case for which there exists at least one node $l \in X_n$ in the network (i.e., node $l$ is neither curious nor following the privacy-preserving protocol). 
Also, $i \in X_c, \ \forall \ i \in X \setminus \{l, j\}$ (i.e., all other nodes in the network are curious). 
At the initialization of our algorithm, node $j$ will inject an initial offset to its value $x_j'[0] = x_j + u_j$ as in \eqref{initial_injection} while $l$ will maintain its value $x_l$, i.e.,  $x_l'[0] = x_l$. 
Let us assume that at the first time step $k=0$, $j$ and $l$ are the selected pair. 
Then, they will update their values as
\begin{eqnarray}\label{init_decompose1}
x_j'[1] = & \big(x_j'[0] + x_l'[0] \big)/2 + u_j[0], \ & \text{and} \nonumber \\
x_l'[1] = & \big(x_j'[0] + x_l'[0] \big)/2. & \;
\end{eqnarray}
When nodes $j$ and $l$ are paired with curious nodes at time steps $k=1, 2, ...$, their values $x_j'[1]$ and $x_l'[1]$ will become known (along with $(x_j'[0] + x_l'[0]) / 2$ and $u_j[0]$). 
Later on, curious nodes could determine $u_j[1]$, $u_j[2]$ and eventually $u_j[L_j] = -u_j - \sum_{t=1}^{L_j - 1} u_j[t]$. 
Thus, they can determine $u_j$ and compute the exact sum of the initial values of nodes $j$ and $l$ as $x_j + x_i = 2x_l'[1] - u_j$ (the value of $x_l'[1]$ becomes known at iteration $k=1$, whereas $u_j$ can be calculated after iteration $L_j$, once $u_j[L_j]$ is available). 
\item[3)] Let us consider the case for which there exists at least one node $l \in X_p$ in the network that also follows the proposed privacy-preserving protocol. 
Also, $i \in X_c, \ \forall \ \in X \setminus \{j, l\}$ (i.e., all other nodes in the network are curious). 
At the initialization of our algorithm, $j$ and $l$ will inject initial offsets to their values $x_j'[0] = x_j + u_j$ and $x_l'[0] = x_l + u_l$, respectively, as in \eqref{initial_injection}. 
During the iteration of our algorithm, each time $j$ or $l$ are paired with another curious node $i \in X_c$ at time steps $k=1, 2, ...$, they will inject offsets to their values $u_j[k]$, and $u_l[k]$ as 
\begin{eqnarray}\label{init_decompose2}
x_j'[k+1] = & \big(x_j'[k] + x_i'[k] \big)/2 + u_j[k], \ & \text{or} \nonumber \\
x_l'[k+1] = & \big(x_l'[k] + x_i'[k] \big)/2 + u_l[k]. & \;
\end{eqnarray}
Let us assume that at time step $k'$, nodes $j$ and $l$ are selected as a pair. 
Both nodes will inject offsets $u_j[k']$ and $u_l[k']$ to their values as
\begin{eqnarray}\label{init_decompose3}
x_j'[k'+1] = & \big(x_j'[k'] + x_l'[k'] \big)/2 + u_j[k'], \ & \text{and} \nonumber \\
x_l'[k'+1] = & \big(x_l'[k'] + x_j'[k'] \big)/2 + u_l[k']. & \;
\end{eqnarray}
The sum of values of $u_j[k'-1]$ and $u_l[k'-1]$, which were injected after the previous exchange of each node and are known only to nodes $j$ and $l$, additively affect both $x_j'[k'+1]$ and $x_i'[k'+1]$. 
As a result, the privacy of both nodes $j$ and $l$ is preserved (curious nodes can only resolve the sum of $u_j[k'-1] + u_l[k'-1]$, but not the individual values). 
As we mentioned in the previous scenario, since only nodes $j$ and $l$ are following the proposed protocol, curious nodes $i \in X_p$ will be able to compute the exact sum of the initial values $x_j + x_l$ of nodes $j$ and $l$. 
\end{itemize}
From the above scenarios, we have that node $j$ is able to preserve the privacy of its initial value if (i) there exists at least one node $l \in X_n$ in the network and node $l$ first exchanges information with node $j$ among them, or (ii) there exists at least one other node $l \in X_p$, $l \neq j$,  that also follows the proposed privacy-preserving protocol. 
In both cases, the privacy of the initial value of node $j \in X_p$ is preserved. \\
This completes the proof.
\end{proof}

\section{Illustrative Example}\label{sec:Example}

In this section, we present an example to illustrate the behavior and potential advantages of our proposed algorithm. 
We analyze the case of a  fully connected graph of $20$ nodes with the initial value $x_i$ for each of the nodes chosen independently and uniformly in $[0, 1]$. 
For this case, each node $i \in X$ executes Algorithm~\ref{algorithm:1} and Algorithm~\ref{algorithm:2}, i.e., it utilizes (i) the privacy-preserving protocol described in Section~\ref{Priv_prot}, and (ii) the distributed stopping protocol described in Section~\ref{stop_prot}. 
Specifically, during the operation of Algorithm~\ref{algorithm:1}, in order to ensure privacy-preservation, each node $i$ infuses pseudo random offsets $u_i[k]$ chosen uniformly in $[-1, 1]$, for $\ k=0,1,...,L_i-2$ where $L_i$ is the number of steps elapsed at the first activation of node $i$ after it has communicated at least once with all the other nodes in the network (note that $u_j[k]$ is zero when node $j$ is not selected). 
For distributed stopping guarantees, each node $i$ initially sets its flag to $0$ and takes the value of $\varepsilon=0.0001$.  

In Fig.~\ref{model3}, the average of the initial states of the nodes is equal to $\bar{x} = 0.3699$. 
We show the value of each node $i$ plotted against the required time steps for convergence. 
We observe that each node calculates the approximate average of the initial states after $2140$ time steps.
When the nodes stop, the error bound between the absolute difference of the value of any node in the system from the exact average value ${\overline x}$ is smaller than $\varepsilon$. 

\begin{figure}[t]
    \centering
    \includegraphics[width=\linewidth]{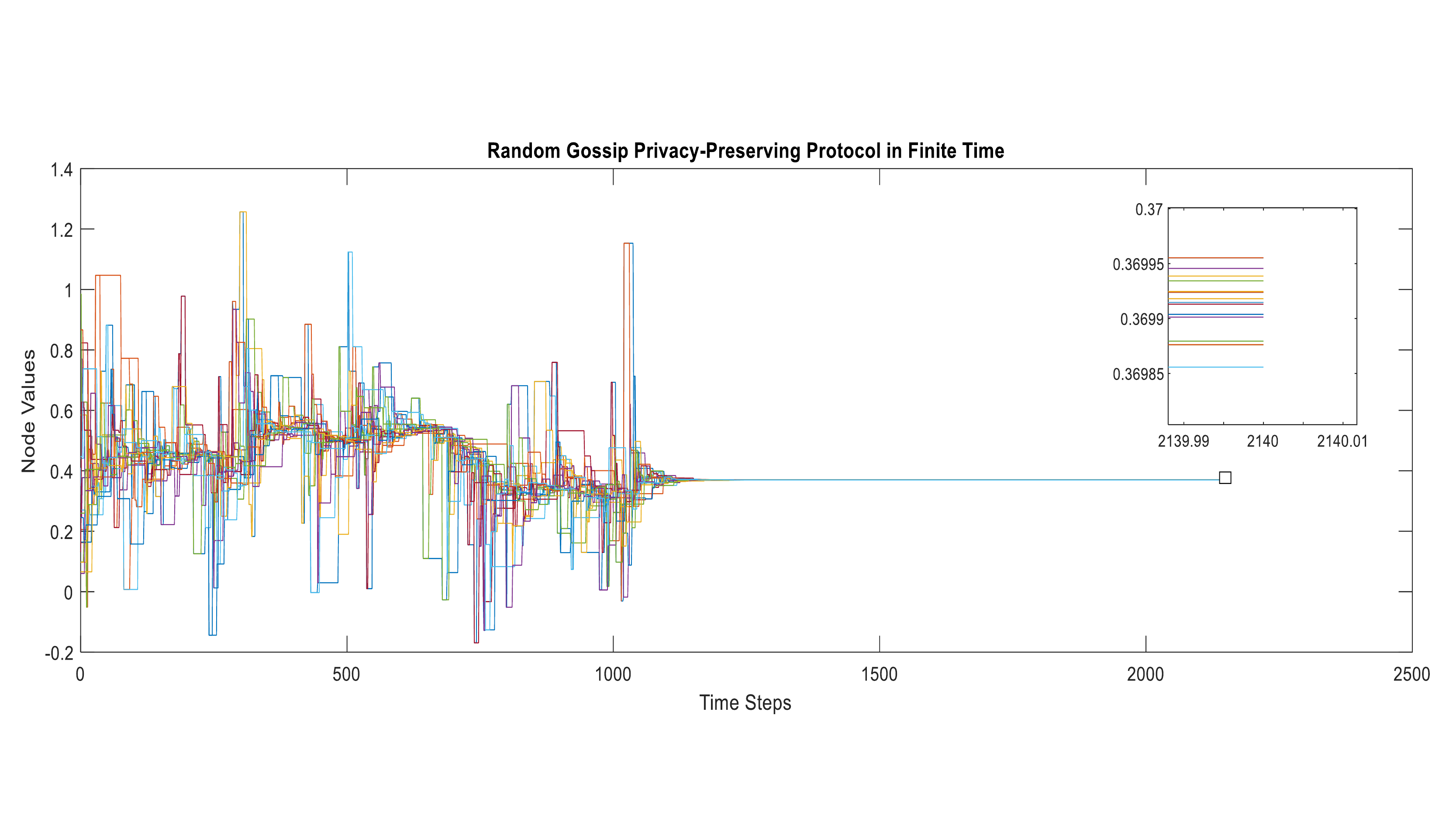}
\caption{Convergence of gossip-based privacy-preserving protocol in finite time.}
\label{model3}
\end {figure}


\addtolength{\textheight}{-3cm}   


\section{CONCLUSIONS AND FUTURE WORK}\label{sec:conlcusion}
In this paper, we have considered the problem of privacy-preserving average consensus in finite time in the presence of curious nodes. 
We presented a distributed privacy-preserving average consensus algorithm that enables all of the components of a distributed system, each with some initial value, to reach (approximate) average consensus on their initial values after executing a finite number of iterations. 
We proposed an enhancement of the standard gossip protocol that can be followed by each node that does not want to reveal its initial value to curious nodes with which it might interact. 
We characterized the conditions on the information exchange that guarantee privacy-preservation for a specific node. 
Furthermore, we also provided a criterion that allows the nodes to determine, in a distributed manner, when to terminate their operation because approximate average consensus has been reached, Finally, we presented simulation results of our proposed protocol.  

In the future we plan to extend the results of this paper to the case where the proposed protocol operates over arbitrary directed graphs (not necessarily fully connected). 
Furthermore, achieving distributed stopping for finite time average consensus with privacy guarantees is to be investigated in our future work.


\bibliographystyle{IEEEtran}        
\bibliography{bibliografia_consensus}           

\begin{thebibliography}{10}
\providecommand{\url}[1]{#1}
\csname url@samestyle\endcsname
\providecommand{\newblock}{\relax}
\providecommand{\bibinfo}[2]{#2}
\providecommand{\BIBentrySTDinterwordspacing}{\spaceskip=0pt\relax}
\providecommand{\BIBentryALTinterwordstretchfactor}{4}
\providecommand{\BIBentryALTinterwordspacing}{\spaceskip=\fontdimen2\font plus
\BIBentryALTinterwordstretchfactor\fontdimen3\font minus
  \fontdimen4\font\relax}
\providecommand{\BIBforeignlanguage}[2]{{%
\expandafter\ifx\csname l@#1\endcsname\relax
\typeout{** WARNING: IEEEtran.bst: No hyphenation pattern has been}%
\typeout{** loaded for the language `#1'. Using the pattern for}%
\typeout{** the default language instead.}%
\else
\language=\csname l@#1\endcsname
\fi
#2}}
\providecommand{\BIBdecl}{\relax}
\BIBdecl

\bibitem{2011:Christoforos_Garcia}
A.~D. {Dom\'{i}nguez-Garc\'{i}a} and C.~N. Hadjicostis, ``Distributed
  algorithms for control of demand response and distributed energy resources,''
  in \emph{Proceedings of $50^{th}$ IEEE Conference on Decision and Control and
  European Control Conference}, 2011, pp. 27--32.

\bibitem{2005:Ren_Beard_Atkins}
W.~Ren, R.~W. Beard, and E.~M. Atkins, ``A survey of consensus problems in
  multi-agent coordination,'' in \emph{Proceedings of American Control
  Conference}, 2005, pp. 1859--1864.

\bibitem{1996:Lynch}
N.~Lynch, \emph{Distributed Algorithms}.\hskip 1em plus 0.5em minus 0.4em\relax
  San Mateo: CA: Morgan Kaufmann Publishers, 1996.

\bibitem{2003:Koetter}
R.~Koetter and M.~M\'edard, ``An algebraic approach to network coding,''
  \emph{IEEE/ACM Transactions on Networking}, vol.~11, no.~5, pp. 782--795,
  October 2003.

\bibitem{2005:Hromkovic}
J.~Hromkovic, R.~Klasing, A.~Pelc, P.~Ruzicka, and W.~Unger,
  \emph{Dissemination of Information in Communication Networks},
  Springer-Verlag, 2005.

\bibitem{2008:Cortes}
J.~Cort\'{e}s, ``Distributed algorithms for reaching consensus on general
  functions,'' \emph{Automatica}, vol.~44, no.~3, pp. 726--737, March 2008.

\bibitem{2008:Sundaram_Hadjicostis}
S.~Sundaram and C.~N. Hadjicostis, ``Distributed function calculation and
  consensus using linear iterative strategies,'' \emph{IEEE Journal on Selected
  Areas in Communications}, vol.~26, no.~4, pp. 650--660, May 2008.

\bibitem{2018:BOOK_Hadj}
C.~N. Hadjicostis, A.~D. {Dom\'{i}nguez-Garc\'{i}a}, and T.~Charalambous,
  ``Distributed averaging and balancing in network systems, with applications
  to coordination and control,'' \emph{Foundations and Trends\textregistered
  ~in Systems and Control}, vol.~5, no. 3--4, 2018.

\bibitem{2010:Dimakis_Rabbat}
A.~G. Dimakis, S.~Kar, J.~M.~F. Moura, M.~G. Rabbat, and A.~Scaglione, ``Gossip
  algorithms for distributed signal processing,'' \emph{Proceedings of the
  IEEE}, vol.~98, no.~11, pp. 1847--1864, November 2010.

\bibitem{Kefayati:2007}
M.~{Kefayati}, M.~S. {Talebi}, B.~H. {Khalaj}, and H.~R. {Rabiee}, ``Secure
  consensus averaging in sensor networks using random offsets,'' in \emph{IEEE
  International Conference on Telecommunications and Malaysia International
  Conference on Communications}, May 2007, pp. 556--560.

\bibitem{Nozari_2017}
E.~Nozari, P.~Tallapragada, and J.~Cortes, ``Differentially private average
  consensus: Obstructions, trade-offs, and optimal algorithm design,''
  \emph{Automatica}, vol.~81, pp. 221--231, 2017.

\bibitem{2016:CortesPappas}
J.~Cort\'{e}s, G.~E. Dullerud, S.~Han, J.~L. Ny, S.~Mitra, and G.~J. Pappas,
  ``Differential privacy in control and network systems,'' in \emph{$55^{th}$
  IEEE Conference on Decision and Control}, 2016, pp. 4252--4272.

\bibitem{2013:Nikolas_Hadj}
N.~Manitara and C.~N. Hadjicostis, ``Privacy-preserving asymptotic average
  consensus,'' in \emph{European Control Conference}, Zurich, 2013, pp.
  760--765.

\bibitem{Mo-Murray:2017}
Y.~{Mo} and R.~M. {Murray}, ``Privacy preserving average consensus,''
  \emph{IEEE Transactions on Automatic Control}, vol.~62, no.~2, pp. 753--765,
  2017.

\bibitem{west}
D.~B. West, \emph{Introduction to Graph Theory}.\hskip 1em plus 0.5em minus
  0.4em\relax Upper Saddle River, New Jersey, 2001.

\bibitem{2020:RikosThemisJohHadj_CDC}
A.~I. Rikos, T.~Charalambous, K.~H. Johansson, and C.~N. Hadjicostis,
  ``Privacy-preserving event-triggered quantized average consensus,'' in
  \emph{Proceedings of $59^{th}$ IEEE Conference on Decision and Control},
  2020, pp. 6246--6253.

\bibitem{2021:Priv_Rikos_Hadj_Joh_Themis_TCNS}
------, ``Distributed event-triggered algorithms for finite-time
  privacy-preserving quantized average consensus,'' \emph{arXiv preprint
  arXiv:2102.06778}, 2021.

\bibitem{2019:Cascudo_Christensen}
Q.~Li, I.~Cascudo, and M.~G. Christensen, ``Privacy-preserving distributed
  average consensus based on additive secret sharing,'' in \emph{European
  Signal Processing Conferenec (EUSIPCO)}, 2019, pp. 1--5.

\bibitem{Hadjicostis_Garcia_2020}
C.~N. Hadjicostis and A.~D. Dominguez-Garcia, ``Privacy-preserving distributed
  averaging via homomorphically encrypted ratio consensus,'' \emph{IEEE
  Transactions on Automatic Control}, vol.~65, no.~9, pp. 3887--3894, 2020.

\bibitem{2019:Ruan_Wang}
M.~Ruan, H.~Gao, and Y.~Wang, ``Secure and privacy-preserving consensus,''
  \emph{IEEE Transactions on Automatic Control}, vol.~64, no.~10, pp.
  4035--4049, 2019.

\bibitem{Hadjicostis_2018}
C.~N. Hadjicostis, ``Privary preserving distributed average consensus via
  homomorphic encryption,'' in \emph{Proceedings of $57^{th}$ IEEE Conference
  on Decision and Control (CDC)}, 2018, pp. 1258--1263.

\bibitem{2019:Charalambous_Manitara_Hadjicostis}
T.~Charalambous, N.~E. Manitara, and C.~N. Hadjicostis, ``Privacy-preserving
  average consensus over digraphs in the presence of time delays,'' in
  \emph{Proceedings of $57^{th}$ Annual Allerton Conference on Communication,
  Control, and Computing}, 2019, pp. 238--245.

\bibitem{2004:XiaoBoyd}
L.~Xiao and S.~Boyd, ``Fast linear iterations for distributed averaging,''
  \emph{Systems and Control Letters}, vol.~53, no.~1, pp. 65--78, September
  2004.

\bibitem{Allerton13_Manitara}
N.~E. Manitara and C.~N. Hadjicostis, ``Distributed stopping in average
  consensus via event-triggered strategies,'' in \emph{Proceedings of $51^{st}$
  Annual Allerton Conference on Communication, Control, and Computing}, 2013,
  pp. 1336--1343.

\end{thebibliography}

\end{document}